\documentclass{amsart}
\usepackage{geometry}
\usepackage{cite}
\newtheorem{thm}{Theorem}[section]
\newtheorem{cor}{Corollary}[section]
\newtheorem{conj}{Conjecture}[section]
\newtheorem{prop}{Proposition}[section]
\newtheorem{lem}{Lemma}[section]

\theoremstyle{definition}

\newtheorem{definition}{Definition}[section]

\newcommand{\lspan}{\operatorname{span}}
\newcommand{\cH}{\mathcal{H}}
\newcommand{\cU}{\mathcal{U}}
\newcommand{\cA}{\mathcal{A}}
\newcommand{\cB}{\mathcal{B}}
\newcommand{\cP}{\mathcal{P}}
\newcommand{\Nset}{\mathbb{N}}

\newcommand{\Rset}{\mathbb{R}}
\newcommand{\Cset}{\mathbb{C}}
\newcommand{\Xl}{{\rm X}_{\lambda}}
\newcommand{\rL}{\mathrm{L}}
\newcommand{\Wr}{\operatorname{Wr}}
\newcommand{\ord}{\operatorname{ord}}
\newcommand{\hT}{\hat{T}}
\newcommand{\hH}{\hat{H}}

\newcommand{\ta}{\tilde{a}}

\begin{document}

\title[]{Rational extensions of the quantum harmonic oscillator and exceptional Hermite polynomials}
\author{David G\'omez-Ullate}\thanks{This work has been partially supported by the Spanish MINECO-FEDER Grants MTM2009-06973, MTM2012-31714, the Catalan Grant
2009SGR--859 and the Canadian NSERC grant RGPIN-228057-2009.} 
\address{ Departamento de F\'isica Te\'orica II, Universidad Complutense de
Madrid, 28040 Madrid, Spain.}
\author{ Yves Grandati}
\address{Equipe BioPhysStat, LCP A2MC, Universit\'e de Lorraine--Site de Metz,
1 Bvd D. F. Arago, F-57070, Metz, France.}
\author{Robert Milson}
\address{Department of Mathematics and Statistics, Dalhousie University,
Halifax, NS, B3H 3J5, Canada.}
\begin{abstract}
 We prove that every rational extension of the quantum harmonic oscillator that is exactly solvable by polynomials is monodromy free, and therefore can be obtained by applying a finite number of state-deleting Darboux transformations on the harmonic oscillator. Equivalently, every exceptional orthogonal polynomial system of Hermite type can be obtained by applying a Darboux-Crum transformation to the classical Hermite polynomials. Exceptional Hermite polynomial systems only exist for even codimension $2m$, and they are indexed by the partitions $\lambda$ of $m$. We provide explicit expressions for their corresponding orthogonality weights and differential operators and a separate proof of their completeness. Exceptional Hermite polynomials satisfy a $2\ell+3$ recurrence relation where $\ell$ is the length of the partition $\lambda$. Explicit expressions for such recurrence relations are given.
 
\end{abstract}
\maketitle

\section{Introduction and main results}

In the seminal paper of Duistermaat and Gr\"unbaum, \cite{Duistermaat1986}, the authors explore the connections between bispectrality of a Schr\"odinger operator, trivial monodromy of the potential and Darboux transformations. They proved that for potentials on the real line that decay at infinity, trivial monodromy of the potential always implies the existence of a bispectral problem and moreover that the potential can be obtained by a finite number of Darboux transformations from that of the free particle. Oblomkov later generalized these results to trivial monodromy potentials with quadratic increase at infinity,\cite{Oblomkov1999}, which are Darboux connected to the harmonic oscillator.

Darboux transformations have long been used as a mechanism to generate new solvable potentials starting from known ones,\cite{INFELD1951,deift,SUKUMAR1985a,MIELNIK1984,Bagrov1995}. However, only a certain subclass of these transformations will preserve the polynomial character of the eigenfunctions, namely those for which the log derivative of the factorizing function is a rational function,\cite{Gomez-Ullate2004d,Gomez-Ullate2004c}. This class can also be described as a set of discrete symmetries that preserve the form of the Riccati-Schr\"odinger equation, \cite{Grandati2011b,Grandati2010}. 

The application of these \textit{rational} Darboux transformations on known exactly solvable potentials leads to transformed potentials that differ from the original ones by the addition of extra terms which are rational functions of a suitable variable and are typically bounded on the domain of definition of the Hamiltonian. The transformed potential therefore has the same asymptotic behaviour as the original one, and it is called a \textit{rational extension} \cite{Grandati2011b,Grandati2011c,Grandati2012b,Marquette2013}.

The eigenfunctions of these transformed Hamiltonians are often expressible in terms of \textit{exceptional orthogonal polynomials}, which are families of orthogonal polynomials with a finite number of gaps in their degree sequence that nevertheless span a complete basis of their corresponding Hilbert spaces.

Exceptional orthogonal polynomials of codimension one were first introduced in \cite{Gomez-Ullate2010a,Gomez-Ullate2009a} as the polynomial eigenfunctions of a Sturm-Liouville problem, extending a famous result by Bochner, \cite{Bochner1929}. 
The connection to Darboux transformations was not yet evident at that stage, and the main result was proved using the classification of normal forms for the flags of univariate polynomials of codimension one, \cite{Gomez-Ullate2005a,Gomez-Ullate2007e}.

Shortly after, Quesne \cite{Quesne2008,Quesne2009b}  showed the relation between exceptional orthogonal polynomials and the Darboux
transformation, which enabled her to
obtain examples of codimension two. Higher-codimensional families were first obtained by
Odake and Sasaki \cite{Odake2009a}. The same authors further showed the
existence of two families of ${\rm X}_m$-Laguerre and ${\rm X}_m$-Jacobi
polynomials \cite{Odake2010b}, the existence of which was further explained in \cite{Gomez-Ullate2010b,Sasaki2010} for ${\rm X}_m$-Laguerre polynomials and in \cite{Gomez-Ullate2012c} for  ${\rm X}_m$-Jacobi polynomials.

The next step towards the complete understanding of exceptional orthogonal polynomials came with the multi-indexed families associated to Darboux-Crum or iterated Darboux transformations \cite{Crum1955}, which were first proposed in \cite{Gomez-Ullate2012a} and later generalized in \cite{Odake2011c,Quesne2011}.

Exceptional orthogonal polynomials have been applied in a number of interesting physical contexts, such as quantum superintegrability \cite{Marquette2013a,Post2012a}, Dirac operators minimally coupled to external fields,\cite{Ho2011a}, entropy measures in quantum information theory, \cite{Dutta2011a},  Schr\"odinger's equation with position-dependent mass, \cite{Midya-Roy}, or discrete quantum mechanics, \cite{Odake2011a}. 

The mathematical properties of these exceptional polynomial families are also the subject of recent study, such as the asymptotic and interlacing properties of their zeros \cite{Gomez-Ullate2013a} and higher order recursion formulas, \cite{Odake2013}. 

The main results of this paper have a double interpretation depending on whether the focus is placed on the potential of the Schr\"odinger operator, or on the eigenfunctions. From the potential point of view, we characterize all rational extensions of the harmonic oscillator that are exactly solvable by polynomials. This is achieved by proving that any operator that is solvable by polynomials necessarily has trivial monodromy. Oblomkov's result \cite{Oblomkov1999} provides the necessary link to Darboux-Crum transformations, and we show which of them lead to a regular potential.
From the point of view of the eigenfunctions, to every such rational extension that is solvable by polynomials there corresponds an associated orthogonal polynomial system of Hermite type. This will be an exceptional polynomial system since the sequence of eigenfunctions does not contain polynomials of all degrees. The results in this paper provide a full classification of exceptional Hermite polynomials, showing that every such family can be obtained by applying a sequence of Darboux transformations to the classical Hermite polynomials. For the Hermite case, this results proves the conjecture made by two of the authors in \cite{Gomez-Ullate2012}.

We would like to recall at this point that some examples of exceptional Hermite polynomials have been treated in the literature. Exceptional Hermite polynomials related to a single step state-adding Darboux transformation were first considered in \cite{dubov1992} and later in \cite{junker1997,carinena2008,fellows2009,Dutta-Roy} and for  2-step transformations in \cite{Bagrov1995}. Likewise,  Wronskian determinants of sequences of Hermite polynomials have been investigated previously by Clarkson, \cite{clarkson1}, who has performed numerical investigations of the position of their zeros in the complex plane, later extended in \cite{Felder2012a}. His motivation was that these functions appear as rational solutions to nonlinear differential equations of Painlev\'e
type, \cite{clarkson2}. In this paper we show that these Wronskians are solutions to \textit{linear} second order differential equations and for suitably chosen sequences, they form an orthogonal polynomial system.

In the rest of this section we introduce some preliminary definitions and we state the main results of the paper.
\begin{definition}
  A quantum Hamiltonian 
  \begin{equation}\label{eq:Hgen}
  \cH=-\partial_{xx} +U(x)
 \end{equation}  
 is said to be \textit{exactly solvable by polynomials}\footnote{We
   note that this notion is equivalent to the concept of a polynomial
   Sturm-Liouville problem (PSLP) introduced in
   \cite{Gomez-Ullate2012}.} if there exist functions
 $\mu(x),\zeta(x)$ such that \emph{for all but finitely many} $k\in
 \Nset$ there exists a degree $k$ polynomial $y_k(z)$ such that
  \[ \psi_k(x) = \mu(x) y_k(\zeta(x)) \]
  is an eigenfunction (in the $L^2$ sense) of $\cH$.  
\end{definition}

\begin{definition}
  A rational extension of the harmonic oscillator is a potential of
  the form
  \begin{equation}
    \label{eq:XHermite}
    U(x) = x^2 + \frac{a(x)}{b(x)}
  \end{equation}
  where $a(x), b(x)$ are real polynomials such that $\deg a < \deg b$.
  If $b(x)\neq 0$ for all real $x$ we will say that the potential is
  regular.
\end{definition}

With these two definitions we are ready to state the main theorem.

\begin{thm}\label{thm:main}
  If $U(x)$ is a regular rational extension of the harmonic oscillator  that is solvable  by polynomials, then $U(x)$ has the form
  \begin{equation}\label{eq:Urat}
  U(x) = x^2 - \partial_{xx} \log \Wr[H_{k_1}, H_{k_1+1}, H_{k_2},
  H_{k_2 +1}, \ldots, H_{k_\ell}, H_{k_\ell+1}] , 
  \end{equation}
   where $\Wr$ is the
  Wronskian operator, $H_n(x)$ is the nth degree Hermite
  polynomial and where
  \begin{equation}\label{eq:regcond}
   k_j+1 < k_{j+1} ,\quad j=1,2,\ldots, \ell-1.
   \end{equation} Conversely, every
  potential of the form shown above is non-singular and solvable by
  polynomials.
\end{thm}

In other words, the theorem states that every potential of the form \eqref{eq:XHermite} that is solvable by polynomials can be obtained from the harmonic oscillator by a sequence of state-deleting Darboux transformations.
The condition on the possible sequences \eqref{eq:regcond} ensures the regularity of the potential \eqref{eq:Urat}.

The following theorem provides an intermediate result that will be used in the proof of Theorem \ref{thm:main}.
\begin{thm}
  \label{thm:trivmonod}
  Every rational extension  of the harmonic oscillator \eqref{eq:Hgen}-\eqref{eq:XHermite} that is solvable
  by polynomials necessarily has trivial monodromy.
\end{thm}

In a certain way, this last theorem provides the converse statement to a result proved by Oblomkov in \cite{Oblomkov1999} (see Theorem \ref{thm:oblomkov}), thus showing that trivial monodromy of a rational extension and exact solvability by polynomials are equivalent.

The paper is organized as follows: In the next Section we introduce some preliminary results and definitions concerning second order differential operators that are solvable by polynomials. Section \ref{sec:Trivial} is devoted to proving Theorem \ref{thm:trivmonod} while Section \ref{sec:Darboux} provides the proof of Theorem \ref{thm:main}, which is based in Theorem \ref{thm:trivmonod} and previous results of  Oblomkov, Krein and Adler. In Section \ref{sec:XHermite} we study the eigenfunctions of these rational extensions, which are exceptional Hermite polynomials. We provide their degree sequences and codimension, orthogonality weights, differential equations, a convenient way to label them in terms of partitions, a proof of their completeness, recursion formulas and square norms.

\section{Preliminaries}\label{sec:Prelim}

\subsection{Differential operators of Hermite type}

Since the eigenfunctions of a Schr\"odinger operator that is exactly solvable by polynomials are polynomials in a suitable variable with a common pre-factor, it will be useful to work with an equivalent operator that has polynomial eigenfunctions. 

Let $T$ be a general second order differential operator
\begin{equation}\label{eq:Tgeneral} 
T[y]=p(z)y'' + q(z) y' + r(z) y
\end{equation}
acting on a function $y=y(z)$. 

\begin{definition}
A polynomial $y(z)$ is said to be a \textit{polynomial
    eigenfunction of $T$} if there exists a  $\lambda\in\Rset$ such that  $T[y] = \lambda y$. A second order differential operator $T[y]$ is \textit{solvable by polynomials} if there exists
    a polynomial eigenfunction of degree $n$ for all but finitely many
    $n\in \Nset$.  
\end{definition}

Without loss of generality we can assume that the
polynomials $y_k(z)$ do not share a common root.  If they did, we
could write every polynomial as 
\[ y_k(z) = \tilde{y}_{k-d}(z) \sigma(z) \] where the
$\tilde{y}_k(z)$ polynomials are relatively prime and where
$\sigma(z)$ is a polynomial of degree $d\geq 1$.  We could then
express the eigenfunctions as
\[ \psi_k(x) = \tilde{\mu}(x)  \tilde{y}_k(\zeta(x)),\quad
\tilde{\mu}(x) = \mu(x) \sigma(\zeta(x)).\]

A simple gauge transformation and a change of variable show that to every Schr\"odinger operator \eqref{eq:Hgen} that is solvable by polynomials there corresponds an operator $T$ with the same property. More specifically, let $\cH$ be one such operator, then the operator $T[y]$ defined by
\begin{subequations}  \label{eq:HTrel}
  \begin{align}
   T[y]\circ \zeta &= -\mu^{-1} \cH[\psi],\\
  y\circ\zeta &= \mu^{-1} \psi,
  \end{align} 
  \end{subequations}
 is solvable by polynomials. The coefficients of $T$ are given by 
 \begin{subequations}\label{eq:pqrdef}
\begin{align}
    \label{eq:pzeta}
     p\circ \zeta&= (\zeta')^2  ,\\
  \label{eq:qzeta}    q\circ \zeta &=  \frac{2\zeta'\mu'}{\mu}  + \zeta'',\\
   \label{eq:rzeta}  r\circ\zeta &=  \frac{\mu''}{\mu} - U.
  \end{align}
  \end{subequations}

\begin{definition}
A solvable by polynomials operator $T[y]$ is said to be \textit{imprimitive}  if there exists a polynomial $\sigma(z)$ of degree $\geq 1$ that divides \emph{every}
polynomial eigenfunction of $T$. Otherwise, we will say that $T$ is \textit{primitive}.

\end{definition}

The next two propositions express simple conditions on the coefficients $p,q,r$ of an operator $T$ that is solvable by polynomials.

\begin{prop}
  \label{prop:cramer}
  Let $T[y] = p(z)y''+q(z) y'+r(z)y$ be a differential operator such that
  \[T[y_i] = g_i,\; i=1,2,3,\]
  where $y_i(z), g_i(z)$ are polynomials with
  $y_1, y_2, y_3$ linearly independent.  Then, $p(z),q(z),$ and $r(z)$ are rational
  functions with  the Wronskian
  $\Wr[y_1,y_2,y_3]$ in the denominator.
\end{prop}
\begin{proof}
  It suffices to apply Cramer's rule to solve the linear system
  \[ \begin{pmatrix}
    y''_1 & y'_1 & y_1\\
    y''_2 & y'_2 & y_2\\
    y''_3 & y'_3 & y_3
  \end{pmatrix}
  \begin{pmatrix}
    p\\q\\r
  \end{pmatrix}
  =
  \begin{pmatrix}
    g_1\\ g_2\\ g_3
  \end{pmatrix}
  \]
\end{proof}

As usual, we define the degree of a rational function as the degree of the numerator
minus the degree of the denominator.
\begin{prop}
  \label{prop:pqr}
  Let $T[y] = p(z) y'' + q(z) y' + r(z) y$ be a primitive
  operator that is solvable by polynomials.  Then, $p(z)$ is a
  polynomial while $q(z)$ and $r(z)$ are rational functions such that
   \[ \deg p \leq 2,\quad \deg q \leq 1,\quad \deg r \leq 0 .\]
\end{prop}
In other words, operator $T$ does not raise the degree of any polynomial on which it acts. We postpone the proof of this last proposition until the next section. We now focus on rational extensions of the harmonic oscillator, which will be the object of our study for this paper.

\begin{prop}\label{prop:HTequiv}
Let $\cH$ be a rational extension of the harmonic oscillator \eqref{eq:Hgen}-\eqref{eq:XHermite} that is solvable by polynomials. Then $\cH$ is equivalent via \eqref{eq:HTrel}-\eqref{eq:pqrdef} to 
\begin{equation}\label{eq:Tp1}
T[y]=y''+q(z) y'+r(z) y.
\end{equation}
\end{prop}
\begin{proof}
  By assumption, $T[y]$ defined by \eqref{eq:HTrel} has polynomial eigenfunctions $y_k(z),$ $\deg
  y_k = k$ for all but finitely many degrees $k\in \Nset$.  By
  adjusting $\mu(x)$, if necessary, no generality is lost if we assume
  that these polynomials do not have a common factor; i.e., that
  $T[y]$ is primitive.  Hence, by Proposition \ref{prop:pqr}, $p(z)$
  is a polynomial of degree $\leq 2$, and $q(z), r(z)$ are rational
  functions. We see that $U(x)$ cannot be a rational function if $\deg p=2$, since by \eqref{eq:pzeta}   $\zeta(x)$ would be a
  transcendental function, and $U(x)$ given by \eqref{eq:rzeta} would be transcendental too. If $\deg p = 1$, then without loss of generality, $p(z) = 4z$
  and $\zeta(x) = x^2$.  In this case, the polynomial eigenfunctions
  have the form $\mu(x) y_k(x^2)$.  If $p(z)=1$ then $\zeta(x) = x$ and the polynomial eigenfunctions have the form $\mu(x) y_k(x)$, so no generality is lost if we assume
that the latter case holds; i.e., that $p(z) = 1$.

\end{proof}

\subsection{Wronskians of Hermite polynomials}

Wronskian determinants of a sequence of Hermite polynomials will play a prominent role in the rest of the paper, so we will introduce some preliminary notation. Given an indexed set of functions $\{f_{k_1},\dots,f_{k_n}\}$ we will denote its Wronskian determinant by 
\[ \Wr[ f_{k_1},\dots,f_{k_n}]=\left| \begin{array}{cccc}
f_{k_1} & f_{k_2}& \cdots& f_{k_n}\\
f'_{k_1} & f'_{k_2}& \cdots& f'_{k_n}\\
\vdots & \vdots &\ddots&\vdots\\
f^{(n-1)}_{k_1} & f^{(n-1)}_{k_2}& \cdots& f^{(n-1)}_{k_n}
\end{array} \right|.
\]
For an increasing sequence of natural numbers $k_1<k_2<\dots<k_n$ we can define the following non-decreasing sequence $1\leq\lambda_1\leq\lambda_2\leq\dots\leq\lambda_n$ by letting $\lambda_i=k_i-i+1$. In this context, \cite{Felder2012a}, it is customary to interpret $\lambda=(\lambda_1,\dots,\lambda_n)$ as a partition with  $|\lambda|=\sum_{i=1}^n \lambda_i$.

For an indexed sequence of functions $\{f_{k_1},\dots,f_{k_n}\}$ defining a partition $\lambda$ we shall use the following shorthand notation
\[ f_\lambda=\Wr[ f_{\lambda_1},f_{\lambda_2-1},\dots,f_{\lambda_n+n-1}]= \Wr[ f_{k_1},\dots,f_{k_n}]\]

Let $f_n$ be a polynomial of degree $n$, then the degree of a Wronskian is easily expressed in terms of the partition:
\begin{equation}\label{eq:degW}
\deg f_\lambda=|\lambda|=\sum_{i=1}^n \lambda_i  
\end{equation}

We also note that that the Wronskian of Hermite polynomials has well defined parity:
\begin{equation}\label{eq:XHparity1}
H_\lambda(-x)=(-1)^{|\lambda|} H_\lambda(x)
\end{equation}


\section{Trivial monodromy}\label{sec:Trivial}

In this section we show that every operator $T$ of the form \eqref{eq:Tp1} that is solvable by polynomials necessarily has trivial monodromy. This property of trivial monodromy is inherited by the Schr\"odinger operator $\cH$, which will allow us to invoke the important result of Oblomkov \cite{Oblomkov1999} that characterizes potentials with trivial monodromy and quadratic growth at infinity in terms of Darboux transformations of the harmonic oscillator.

\begin{definition}
  Consider a second-order operator $T[y] = y'' + q(z) y' + r(z) y$
  where $q(z), r(z)$ are meromorphic functions.
We say that $\zeta\in \Cset$ is a pole of the operator  if either $q(z)$ or $r(z)$ have a pole at $z=\zeta$.  Moreover, $\zeta$ is a regular singular
    point of the operator if
    \[ \ord_\zeta q \geq -1,\quad \text{and}\quad \ord_\zeta r \geq
    -2.\]
    The order of $\zeta$ as a pole of $T$ is defined as
    \[\ord_\zeta T =
    \min\{\,-1+\ord_\zeta \!q,\,\ord_\zeta\! r\}\]
\end{definition}

\begin{definition}

An operator $T[y]$ has \emph{trivial monodromy} at $\zeta\in
    \Cset$ if for every $\lambda \in \Cset$ there are two linearly
    independent solutions of the differential equation
    \[ T[\phi] = \lambda \phi \] that are meromorphic at $z=\zeta$.
    An operator $T$ has trivial monodromy if the above condition holds for every $\zeta\in\Cset$.
\end{definition}

We first observe that at every pole of $T$ we can decompose the operator into a sum of degree homogeneous parts.
Consider the Laurent expansions of $q(z)$ and $r(z)$ at $z=\zeta$
\begin{subequations}\label{eq:qr}
    \begin{align}
      q(z) &= \sum_{j\geq
        \ord_\zeta\! q} q_j (z-\zeta)^j,\\
      r(z) &= \sum_{j\geq \ord_\zeta\!  r} r_j (z-\zeta)^j,
    \end{align}
    \end{subequations}
    and define the operators
    \begin{subequations}\label{eq:Tj}
    \begin{align}
      T_{j}[y] &= q_{j+1} (z-\zeta)^{j+1} y' +r_j (z-\zeta)^{j} y ,\qquad j\neq -2   \\
          T_{-2}[y] &=y'' +q_{-1} (z-\zeta)^{-1} 
        y' +r_{-2} (z-\zeta)^{-2} y
  \end{align}
  \end{subequations}
   It is clear that the action of $T$ on a function $\phi(z)$ that is meromorphic at $z=\zeta$ can be written as 
    \[ T[\phi] = \sum_{j\geq \ord_\zeta T} T_j[\phi].\] For this reason, we
    will call
    \begin{equation}\label{eq:TLaurent}
    T = \sum_{j\geq \ord_\zeta T} T_j 
    \end{equation}
     the Laurent expansion of $T$
    at $z=\zeta$.

\begin{definition}Let $\cU(T)$ denote the vector space spanned by the polynomial
    eigenfunctions of $T$.
 We define the \textit{order sequence of $T$ at
    $z=\zeta$} as
    \[
    I_{\zeta}(T)  = \{ \ord_\zeta y: y \in \cU(T) \}.
    \]    
\end{definition} 
We note that an equivalent definition of the order sequence is that we can construct a basis of $\cU(T)$ with polynomials 
    \[ y_k(z) = (z-\zeta)^k + \text{higher order terms},\quad k\in I_{\zeta}(T).\] 
We adopt the convention that  $\ord_\infty y = \deg y$ and we can therefore define the following sequence:
\begin{definition}
We define the \textit{degree sequence} of an operator $T$ to be
    \[ I_\infty(T) = \{ \deg y : y\ \in \cU(T) \}.\]
\end{definition}
 In this case, we are considering a basis of $\cU(T)$ consisting of
    polynomials of the form
    \[ y_k(z) = z^k + \text{lower degree terms}, \quad k\in I_\infty(T).\]

For $\zeta \in \Cset$, let $\nu_\zeta$ be the cardinality of
    $\Nset/I_\zeta(T)$, that is the number of gaps in the order sequence.  Let $\nu_\infty$ be the number of gaps in the degree sequence. It is clear that $T$ is solvable by polynomials if and only if the \textit{codimension} $\nu_\infty<\infty$.

\begin{prop}
  \label{prop:odgaps}
 Let $T[y]$ be a differential operator that is solvable by polynomials. Then we have
  $\nu_\zeta\leq \nu_\infty$ for all $\zeta\in \Cset$.  In other words, the number of gaps in the order sequence cannot exceed the number of gaps in the degree sequence.
\end{prop}

\begin{proof}
  Let $\cU_n\subset \cU(T)$ be the subspace spanned by polynomial
  eigenvalues of degree $\leq n$ and let
  \begin{align*}
    I_{\zeta,n} &= \{ \ord_\zeta y : y\in \cU_n \},\qquad
   \nu_{\zeta,n} = n+1 - |I_{\zeta,n}|.
  \end{align*}
  If $j\notin I_{\zeta}$, then $j\notin I_{\zeta,n}$ for all $n$ and $j$ is a persistent gap.  There could also be non-persistent
  gaps, such that $j\notin I_{\zeta,n},\; n\geq j$ but $j\in I_\zeta$.  These
  non-persistent gaps disappear for $m$ sufficiently large.  Observe that the cardinality of
  $I_{\zeta,n}$ is equal to the dimension of $\cU_n$. Hence,
  $\nu_{\zeta,n} = \nu_{\infty,n}$. By assumption, $\nu_{\infty,n} =
  \nu_\infty$ for sufficiently large $n$.  We conclude that $\nu_\zeta$, the
  number of persistent gaps has to be bounded
  by $\nu_\infty$.
\end{proof}


Consider the Laurent decomposition \eqref{eq:qr}-\eqref{eq:Tj}-\eqref{eq:TLaurent} of an operator $T$ at a pole $\zeta\in\Cset$. The following lemma shows that if $T$ is solvable by polynomials, then $\zeta$ is a singular regular point and the order sequence has a well defined structure.

\begin{lem}
  \label{lem:ordergaps}
  Let $T[y]= y''+q(z) y' + r(z) y$ be a primitive operator that is solvable by polynomials.  Then every pole of $T$ is  a regular singular point.  Furthermore, if $\zeta\in \Cset$ is  a pole,
  then $\nu_\zeta\geq 1$ and the order sequence is
  \begin{equation}
    \label{eq:Ibform}
    I_\zeta(T) = \{ 0,2,4,\ldots, 2\nu_\zeta, 2\nu_\zeta+1,2\nu_\zeta+2,\ldots \} = \Nset/
    \{ 1,3,5,\dots 2\nu_\zeta-1 \}.
  \end{equation}
  Moreover, the leading term of $T$ in the Laurent expansion
  \eqref{eq:TLaurent} is
   \[ T_{-2} = \partial_{zz} -
    \frac{2\nu_{\zeta}}{z-\zeta} \partial_z\]
\end{lem}
\begin{proof} Without loss of generality, we take $\zeta = 0$ and
  write $\nu = \nu_0$.  Let
  \[T = \sum_{j=d}^\infty T_j,\quad d= \ord_0 T \] be the Laurent
  expansion at $z=0$.  To prove that $\zeta=0$ is a regular singular
  point it suffices to show that $d\geq -2$.  We observe that each
  $T_j$ is degree-homogeneous, i.e. $T_j$ either annihilates a given
  monomial $z^k$, or it shifts its degree by $j$.  A non-zero $T_j$
  can annihilate at most two distinct monomials.
   $T_{d}$ is the leading term of the operator, so it must preserve the monomial vector space spanned by $z^j, \; j \in I_0$.  Since $T$ is
  primitive, we must have $0\in I_0$, which means that $T_{d}[1] = 0$ and therefore $r_d = 0$.  But $T$ has a pole at $z=0$ so necessarily $d<0$.  If $d=-1$, then $r_{-1} = 0$, contrary to hypothesis that there is a pole at $z=0$, so $d$ must be  $d\leq -2$. We next show that $d=-2$ exactly.  By Proposition
  \ref{prop:odgaps}, $\nu<\infty$.  Let $j\notin I_0$ be one such
  gap. Then, either $j-d\notin I_0$, or $T_{d}$ annihilates $z^{j-d}$.
 We conclude that $1\notin I_0$ must be a gap of $I_0$ since otherwise $T_{d}$ would be required to
annihilate three monomials: $z^0,z^1$ and at least one higher degree monomial, which is impossible.  Thus, for some integer $\alpha \geq
  1$, there exist gaps $1,1-d,1-2d,\ldots, 1-(\alpha-1)d\notin I_0$,
  with $T_{d}[z^{1-\alpha d}]=0$.  Since $T_{d}$ annihilates $1$ and
  $z^{1-\alpha d}$, it cannot annihilate any other
  monomial. Therefore, the above gaps are the only gaps in $I_0$.  It
  follows that $\alpha = \nu$ and that $2\in I_0$ is \emph{not} a
  gap. If the leading order was $d<-2$ then $T_{d}$ would also be
  required to annihilate three monomials: $1,z^2, z^{1-d\alpha}$; which is
  impossible.  We conclude then that $d=-2$ and therefore
  \[ T_{-2} [1] = T_{-2}[z^{2\nu+1}] =0,\] which in turn imply that 
  \[ T_{-2} = \partial_{zz} -
    \frac{2\nu}{z} \partial_z.\]
\end{proof}

\begin{proof}[Proof of Proposition \ref{prop:pqr}]
  By Proposition \ref{prop:cramer}, the operator coefficients are
  rational functions. Now it is obvious that $p(z)$ must be a
  polynomial, since if it had a pole in $\zeta$ then $\ord_\zeta T\geq
  3$ which is forbidden by Lemma \ref{lem:ordergaps}.  Next, let
  $\zeta_i, \; i=1,\ldots, N$ be the poles of the operator.  Consider
  the partial fraction expansions
  \begin{align*}
    q(z) &= q_0(z) + \sum_{i=1}^N
    \frac{q^{(i)}(z)}{(z-\zeta_i)^{a_i}},\quad a_i = \ord_{\zeta_i} q,\\
    r(z) &= r_0(z) + \sum_{i=1}^N \frac{r^{(i)}(z)}{(z-\zeta_i)^{b_i}},
    \quad b_i = \ord_{\zeta_i} r,
  \end{align*}

  where $q^{(i)}(z), r^{(i)}(z)$ are polynomials such that $\deg q^{(i)}< a_i,
  \deg r^{(i)} < b_i$.  Extend the expansion to the operator by setting
  \begin{align*}
    T^{(0)}[y] &= p(z) y'' + q^{(0)}(z) y' + r^{(0)}(z) y,\\
    T^{(i)}[y] &= \frac{q^{(i)}(z)}{(z-\zeta_i)^{a_i}} y' + \frac{r^{(i)} (z)}{(z-\zeta_i)^{b_i}} y,\quad i=1,\ldots, N
  \end{align*}
  In this way
  \[ T = \sum_{i=0}^N T^{(i)} .\] Next, observe that if $y_k(z),\;
  \deg y_k = k$ is a polynomial eigenfunction of $T[y]$, then
  $T^{(i)}[y_k]$ must be a polynomial of degree strictly less than $k$ for
  every $i=1,\ldots, N$.  It follows that $T^{(0)}[y_k]$ is a polynomial
  of degree $k$ for infinitely many $k$.  Therefore, $T^{(0)}[y]$  does not raise degree, and it follows that
\[ \deg p \leq 2,  \quad \deg q^{(0)}\leq 1, \quad  \deg r^{(0)} = 0. \]

\end{proof}

\begin{prop}
  \label{lem:trivmonod0}
  Let $T[y]= y''+q(z) y' + r(z) y$ be a primitive operator that is
  solvable by polynomials.  Then the general solution of the equation $T[y]=\lambda y$ is an entire function for any $\lambda\in\Cset$. In particular, $T$ has trivial monodromy.

\end{prop}


\begin{proof}
  Without loss of generality we suppose that $\zeta=0$ and write $\nu=\nu_0$. By Lemma \ref{lem:ordergaps}, $\nu\geq 1$ and
  \begin{equation}
    \label{eq:Tord-2}
    T= \sum_{j=-2}^\infty T_j,\quad\text{where} \quad T_{-2}
    = \partial_{zz} - \frac{2\nu}{z} \partial_z 
  \end{equation}
  is the Laurent expansion of the operator at $z=0$. For every
  $\lambda \in \Cset$, the roots
  of the indicial equation for the differential equation
  \begin{equation}
    \label{eq:Tlambday}
     (T-\lambda)[y] = 0
  \end{equation}
  are $0$ and $2\nu+1$.  Hence there exists a unique holomorphic
  solution of the form
  \[ y=a(z;\lambda) = z^{2\nu+1} \left( 1+ \sum_{n=1}^\infty a_n z^n
  \right)\] 

  The difference between the roots of the indicial equation of
  \eqref{eq:Tlambday} is an integer, so we have to show that a logarithmic singularity does not arise for the solution
  corresponding to the smaller root.  To that end, choose a basis $y_k
  \in \cU(T), \; k\in I_0(T)$ such that
  \[ y_k = z^k + \text{higher order terms},\quad k\in I_0(T).\] 
  By Lemma \ref{lem:ordergaps},
  \[ I_0(T) = \{ 0,2,\ldots, 2\nu\} \cup \{ 2\nu+1, 2\nu+2,\ldots \}
  .\] By the usual method of Frobenius, there exists a unique series solution
  $y=b(z;\lambda)$ to \eqref{eq:Tlambday} of the form
  \[ b(z;\lambda) = 1+ \sum_{j=1}^{2\nu} b_j z^j + c\,
  a(z;\lambda)\, \log z+ \sum_{j=1}^\infty a_j z^{2\nu+1+j} \]
  Our claim will be proven once we show that  $c=0$; i.e. that
  $b(z;\lambda)$ is holomorphic in $z$.
  Since
  $\cU(T)$ is $T$-invariant, the action of $T$ on the first $\nu+1$
  basis elements can be expressed as
  \begin{align*}
    T[y_0] &= \sum_{k=0}^\nu B_{0k} \,y_{2k} + \sum_{k=2\nu+1}^{N_0}
    A_{0k} z^k,\\
    T[y_{2j}] &= 2j(2j-1-2\nu)y_{2j-2} + \sum_{k=j}^\nu B_{jk}
    \,y_{2k} + \sum_{k=2\nu+1}^{N_j} 
  A_{jk} z^k,\quad j=1,2,\ldots, \nu,
  \end{align*}
  where $A_{jk}, B_{jk}\in \Cset$ and $N_j\in \Nset$.  
  Let us define the polynomial
  \[ p(z;\lambda) = y_0(z) + \sum_{j=1}^\nu p_j(\lambda) y_{2j}(z) \] where the
  coefficients $p_j,\; j=1,\ldots, \nu$ are defined by the recurrence
  relations
  \begin{align*}
    & 2(1-2\nu) p_1 + (B_{00}-\lambda)  = 0\\
    & 4(3-2\nu) p_2 + (B_{11}-\lambda)p_1 + B_{01}  = 0\\
    & 6(5-2\nu) p_3 + (B_{22}-\lambda)p_2 +B_{12} p_1 +  B_{02}  = 0\\
    &\qquad \vdots \\
    & -2\nu p_\nu  + (B_{\nu-1,\nu-1} -\lambda) p_{\nu-1} +
    \sum_{j=1}^{\nu-2} B_{j,\nu-1} b_j + B_{0,\nu-1} = 0
  \end{align*}
  By construction, it follows that
  \[ (T-\lambda)[p] = O(z^{2\nu}).\]
  Hence,
  \[  p(z;\lambda) = 1+ \sum_{j=1}^{2\nu} b_j z^j + O(z^{2\nu+1})\]
  and then
  \[ b(z;\lambda) = p(z;\lambda) + c\, a(z;\lambda) \log z +
  \sum_{j=0}^\infty \ta_j z^{2\nu+1+j} \] for some choice of constants
  $\ta_j \in \Cset$.  
  Next, observe that
  \begin{align*}
      T_{-2}[z^{2\nu+1}\log z] &= (1+2\nu) z^{2\nu-1},\\
      T_{-2}[z^{2\nu+1}] &= 0\\
     T_{-2} [z^{2\nu+1+j}] &= j (j+1+2\nu) z^{2\nu+j-1} =
  O(z^{2\nu}),\quad j\geq 1.
  \end{align*}
 It follows immediately that $c=0$ and
 $b(z;\lambda)$ is holomorphic.  
  The key element in the proof is the fact  that the invariance of $\cU(T)$ guarantees the
  absence of a logarithmic singularity.
\end{proof}



We are now ready to give the proof of Theorem \ref{thm:trivmonod}.
\begin{proof}[Proof of Theorem \ref{thm:trivmonod}]

By Proposition \ref{prop:HTequiv}, every such Schr\"odinger operator $\cH$ is equivalent via \eqref{eq:HTrel} to an operator $T$ of the form \eqref{eq:Tp1} which is solvable by polynomials. Proposition \ref{lem:trivmonod0} then asserts that $T[y]$ has trivial monodromy. We only need to prove that this property is inherited by $\cH$ via the equivalence formulas \eqref{eq:HTrel}.
Indeed, operator $T$ has the form
  \[ T[y] = y'' + q(z) y' + r(z) y \] where $q(z), r(z)$ are rational
  functions.  By Lemma
  \ref{lem:ordergaps},
  \[ q(z)=q_0+ q_1 z -2 \sum_{i=1}^N \frac{\nu_{\zeta_i}}{z-\zeta_i}\]
  where $\zeta_i,\; i=1,\ldots, N$ are the poles of the operator and
  where the $\nu_{\zeta_i} \geq 1$ are the corresponding gap counts.
  Performing an affine change of variable, without loss of generality we can take
  $q_0=0, q_1 = -2$.  Expression \eqref{eq:pzeta} implies that $\zeta(x)=x$ and \eqref{eq:qzeta} that
    \begin{equation}
      \label{eq:mutrivmonod}
      \mu(x) = e^{-x^2} \prod_{i=1}^N (x-\zeta_i)^{-\nu_{\zeta_i}}.
    \end{equation}
    Finally, since $\mu(x)$ is meromorphic and $T$ has trivial monodromy, we see from \eqref{eq:HTrel} that $\cH$ also has trivial monodromy.
\end{proof}

\section{Darboux transformations and regularity of the potential}\label{sec:Darboux}

We have established in the previous sections that a rational potential with quadratic growth at infinity which is solvable by polynomials has trivial monodromy. The connection with Darboux transformations is provided by the following result proved by Oblomkov in \cite{Oblomkov1999}

\begin{thm}[Oblomkov 1999]\label{thm:oblomkov}
Every monodromy-free Schr\"odinger operator $\cH$ with a quadratically increasing rational potential has the form
\begin{equation}\label{eq:HDarb}
\cH=-D_{xx} +x^2 -2 D_{xx}\log \Wr \left[H_{k_1},\dots,H_{k_n}\right]
\end{equation}
where $k_1<\dots<k_n$ is an increasing sequence of positive integers and $H_m(x)$ is the $m$-th Hermite polynomial.
\end{thm}

This expression translates the fact that such a potential can be obtained by a sequence of state-deleting Darboux transformations at levels $k_1,\dots,k_n$ from the quantum harmonic oscillator.

Although not discussed in \cite{Oblomkov1999}, such a potential will in general have singularities in the real line. We need to investigate thus for which sequences $k_1,\dots,k_n$ will the Wronskian $\Wr(H_{k_1},\dots,H_{k_n})$ have no real zeros.

This problem has been addressed by Krein for Sturm-Liouville problems on the half-line, \cite{Krein1957} and independently by Adler \cite{Adler1994} for Sturm-Liouville problems on a bounded interval. Their result is therefore more general than the case we need here and the proof in \cite{Adler1994} can be extended without difficulty to the case of an infinite interval. The main result is the following

\begin{thm}[Krein-Adler]\label{thm:krein}
Let $\phi_j$ be the eigenfunctions of a pure-point Sturm-Liouville operator \mbox{$L=-D_{xx}+U$} defined on the real line 
\begin{equation}
L[\phi_j]=\lambda_j \phi_j,\quad j=0,1,2,\dots
\end{equation}
The Wronskian $\Wr[\phi_{k_1},\dots,\phi_{k_n}]$ has no zeros in the real line if and only if the sequence $\{k_1,\dots,k_n\}$ has the following structure 
\begin{equation}\label{eq:Adlerseq}
\{0,\dots,k_0+1\}\cup\{k_1,k_1+1,k_2,k_2+1,\dots,k_m,k_m+1\}
\quad\text{ with }\quad k_i+1<k_{i+1}
\end{equation}
for all $i=0,\dots,n-1$.
\end{thm}
In other words, the sequence is allowed to have a block of arbitrary length of consecutive integers starting at $0$, followed by any number of blocks of even length. 


It is well known that a sequence of state-deleting Darboux transformations where at each step the ground state is deleted does not introduce singularities in the potential, since the ground state wavefunction at each step has no zeros. The Krein-Adler sequence expresses the fact that the resulting potential in a Darboux-Crum transformation can be regular even if some of the intermediate steps produce singular potentials. Theorem \ref{thm:krein} characterizes precisely those cases in which such a thing occurs.

\begin{proof}[Proof of Theorem 1.1]

By Theorem \ref{thm:trivmonod} every rational extension of the harmonic oscillator has trivial monodromy. Theorem \ref{thm:oblomkov} by Oblomkov implies then that the potential can be obtained by Darboux transformations form the harmonic oscillator, and therefore has the form \eqref{eq:HDarb}. 
Theorem \ref{thm:krein} states that only if the sequence is chosen as in \eqref{eq:Urat} is the potential ensured to be regular.
We note that since the harmonic oscillator is shape invariant, the first block of Darboux transformations $\{0,\dots,k_0+1\}$ in \eqref{eq:Adlerseq} will not lead to rational extensions, but merely to an overall shift in the spectrum. For this reason, that block can be safely neglected in the sequence \eqref{eq:Urat}.

\end{proof}

\section{Exceptional Hermite polynomials}\label{sec:XHermite}

In much the same way as Hermite polynomials are directly related to the harmonic oscillator, it is natural to describe the polynomial eigenfunctions of the rational extensions \eqref{eq:Urat}.

Since the potential \eqref{eq:Urat} is bounding at infinity and Schr\"odinger's equation has by hypothesis an infinite number of eigenfunctions of the form $\mu(x)y_k(x)$ where $y_k(x)$ are polynomials, the $y_k(x)$ will form a complete set of orthogonal polynomials that satisfy a differential equation. This set clearly does not contain polynomials of every degree, so we are dealing with a system of exceptional orthogonal polynomials, that we shall denote by \textit{exceptional Hermite polynomials}. In the rest of the section we introduce the precise definition of such families, we provide expressions for the differential operator and orthogonality weight and we list some examples.

Exceptional Hermite polynomials are defined through Wronskian determinants of a sequence of ordinary Hermite polynomials. Only when that sequence is chosen according to certain rules will the resulting set be an orthogonal polynomial system with a positive definite weight. However, some of the properties of exceptional Hermite polynomials are derived from identities between arbitrary Wronskian determinants. We choose to present these general identities first and to treat the orthogonal polynomial families as a special case, later in the section.

\subsection{Wronskians of Hermite polynomials}

We denote by $\lambda=(\lambda_1,\dots,\lambda_l)$ a non-decreasing sequence of
non-negative integers
\[0\leq  \lambda_1\leq \lambda_2 \leq \cdots \leq \lambda_\ell.\]
Every such partition determines an increasing sequence of non-negative integers
\begin{equation}\label{eq:klambda}
0\leq k_1<k_2<\cdots <k_\ell,\quad\text{where}\quad
k_i=\lambda_i+i-1,\quad i=1,\dots,\ell 
\end{equation}
 which we will refer to as
the \textit{gap sequence}, for reasons that will be shortly evident. 
Conversely, every such gap sequence
determines a partition as per the above formula.\,\footnote{
Note that contrary
to the standard convention, our definition of a partition allows a
string of initial zeros.  This allows us to describe gap sequences
that begin with $0$; for example
\mbox{$ \lambda = (0,0,0,1,1,2,2)$} corresponds to the gap sequence
$0,1,2,4,5,7,8$. } 
Following \cite{Felder2012a}, for any given partition $\lambda=(\lambda_1\dots,\lambda_\ell)$ of length $\ell$  we define $\lambda^2$ to be the \textit{double partition} of length $2\ell$ as
\begin{equation}\label{eq:double}
 \lambda^2 = (\lambda_1,\lambda_1, \lambda_2,\lambda_2,\ldots,
\lambda_\ell, \lambda_\ell),
\end{equation}

We define an \textit{Adler partition} to be either a double partition or a
double partition preceded by an initial  string of zeros of arbitrary length, and a \emph{reduced partition} to be one that begins with a positive $\lambda_1>0$. For a given partition $\lambda$, consider the following Wronskian determinants
\begin{subequations}\label{eq:Hdefs}
\begin{align}
  H_\lambda &:= \Wr[H_{k_1},\dots,H_{k_\ell}]\\
  \label{eq:Hlk} H_{\lambda,j} &:= \Wr[H_{k_1},\dots,H_{k_\ell},H_j],\quad j\notin \{
  k_1,\ldots, k_\ell\}\\
  \hH_{\lambda,i} &:=\Wr[H_{k_1},\ldots
  \widehat{H_{k_i}},\ldots,H_{k_\ell}],\quad i\in {1,2,\ldots,\ell}, \label{eq:Hhat}
\end{align}
\end{subequations}
where the symbol $ \widehat{H_{k_i}}$ means that polynomial is missing in the sequence.
\begin{definition}\label{def:XHermite}
For a given double partition $\lambda^2=(\lambda_1,\lambda_1,\dots,\lambda_\ell,\lambda_\ell)$
  we define the $\Xl$-Hermite polynomials $H^{(\lambda)}_j$ as the following numerable sequence
\begin{equation}
  H^{(\lambda)}_j=H_{\lambda^2,j}, \qquad j\in\mathbb N\backslash\{k_1,k_1+1,\dots,k_\ell,k_\ell+1\}
\end{equation}
\end{definition}
From the above definition and \eqref{eq:degW} it is clear that 
\begin{equation}
\deg H^{(\lambda)}_j(x)= 2|\lambda|-2\ell+j,
\end{equation}
and also that the $\Xl$-Hermite polynomials have well defined parity
\[ H^{(\lambda)}_j(-x)=(-1)^j H^{(\lambda)}_j(x)  \]
For example, if $\lambda=(1,3)$, then \[H^{(1,3)}_j=H_{(1,1,3,3),j}=\Wr(H_1,H_2,H_5,H_6,H_j),\qquad j\in\Nset\backslash\{1,2,5,6\}\]
and we have $\ell=2$, $|\lambda|=4$ and $\deg H^{(1,3)}_j=4+j$, as it is clear from the Wronskian determinant.

In the rest of the section, we shall use the notation $H_\lambda$ and $H_{\lambda,j}$ for general Wronskian determinants \eqref{eq:Hdefs} and we will reserve the notation $H^{(\lambda)}_j$ for the  $\Xl$-Hermite polynomials that form an orthogonal polynomial system.

\subsection{Differential operators and factorization}

Given a partition $\lambda$ corresponding to the gap sequence $\{
k_1,\ldots, k_\ell\}$, we define the following differential operators
of order $\ell$
\begin{align}  
  \cA_\lambda[y]&:=\Wr[H_{k_1},\dots,H_{k_\ell},y]\\
  \cB_\lambda[y] &:= e^{x^2} H_\lambda(x)^{-\ell} \Wr\big[\hH_{\lambda,1},\ldots, \hH_{\lambda,\ell},e^{-x^2} y\big]  
\end{align} 
where the hat in the definition of $\hH_{\lambda,i}$ indicates deletion, as defined in \eqref{eq:Hhat}.
Let us also define the second order differential operators
  \begin{align}
  \label{eq:Tlambdadef}
  T_\lambda[y]&:=y''-2\left(x+\frac{H_\lambda'}{H_\lambda} \right)
  y'+ \left(\frac{H_\lambda''}{H_\lambda}+ 2 x
    \frac{H_\lambda'}{H_\lambda} \right) y\\
  T[y] &:= y'' - 2x y'
\end{align}
where the latter is the usual Hermite operator.  The second order operator \eqref{eq:Tlambdadef} will generally have singular rational coefficients when $\lambda$ is an arbitrary partition. If $\lambda$ is an Adler partition then  $T_\lambda$ will be nonsingular in $\mathbb R$ and we will say it is an $\Xl$-Hermite operator.

\begin{prop}\label{prop:intertwining}
The classical Hermite operator $T$ and $T_\lambda$ obey the following intertwining relations:
\begin{align}
  &(T_{\lambda}-2\ell) \cA_{\lambda} = \cA_{\lambda} T\\
  &\cB_{\lambda}(T_{\lambda}-2\ell) = T \cB_{\lambda}
\end{align}
\end{prop}
\begin{proof}
The above relations can be derived by breaking them down into a
sequence of 1st order intertwining relations as follows.  Define the
first order operators
\begin{align}
  A_{\lambda, k}[y] &:= \frac{\Wr[H_{\lambda, k},y]}{H_{\lambda}}\\
  B_{\lambda, k}[y] &:= \frac{H_\lambda}{H_{\lambda, k}}\left( y' -
    \left( 2x + \frac{H_\lambda'}{H_\lambda}\right) \right)
\end{align}
Writing $T_\lambda = T_{k_1\ldots k_\ell}$ and $T_{\lambda, k} =
T_{k_1\ldots k_\ell, k}$ we have
\begin{align}
  B_i A_i &= T + 2i \\
  A_i B_i &= T_i +2i-2\\
  B_{ij} A_{ij} &= T_i +2j-2\\
  A_{ij} B_{ij} &= T_{ij} +2j-4
\end{align}
and more generally,
\begin{align}
  B_{\lambda, k} A_{\lambda, k} &= T_\lambda + 2k-2\ell\\
  A_{\lambda, k} B_{\lambda, k} &= T_{\lambda, k} + 2k-2\ell-2
\end{align}
\end{proof}

\begin{prop}
  \label{prop:Tlambdaeigenvalue}
  For every partition $\lambda$ we have
  \begin{equation}
    \label{eq:Tlambdaeigenvalue}
    T_{\lambda}[H_{\lambda, j}] =2 (\ell-j) H_{\lambda, j},\quad
    j\notin \{k_1,\ldots, k_\ell\}.
  \end{equation}
where $\ell$ is the length of the partition $\lambda$.
\end{prop}
\begin{proof}
  The proof is immediate from Proposition \ref{prop:intertwining} and
  the fact that \mbox{$H_{\lambda, j}=\cA_{\lambda}[H_j]$}.
\end{proof}
Adapting the above Proposition to the case of Adler partitions we have the following corollary
\begin{cor}
The $\Xl$-Hermite polynomials $H^{(\lambda)}_j$ defined in Definition \ref{def:XHermite} are
eigenfunctions of the following second order differential operator 
\begin{equation}
  T_{\lambda^2}[H^{(\lambda)}_j] = (4\ell-2j) H^{(\lambda)}_j,\quad
  j\in \mathbb N\backslash \{k_1,k_1+1,\ldots, k_\ell,k_\ell+1\}.
\end{equation}
where $T_{\lambda}$ is given by \eqref{eq:Tlambdadef}.
\end{cor}

\subsection{The exceptional subspace}

Let $\cP$ denote the infinite-dimensional space of univariate
polynomials.  Given a partition $\lambda$, with corresponding gap
sequence $k_1,\ldots, k_\ell$, consider the polynomial subspace
\begin{equation}
  \cU_\lambda = \lspan \{ H_{\lambda, j} \colon j\notin \{ k_1,\ldots,
  k_\ell\}\} =  \{ \Wr[H_{k_1},\ldots, H_{k_\ell}, p] \colon p \in \cP
  \}.  
\end{equation}
We call  the set of integers
\[ I_\lambda = \{ \deg p \colon p\in \cU_\lambda \} \] the degree
sequence of the above polynomial subspace.
\begin{prop}
  The codimension of $\cU_\lambda$ in $\cP$ is $|\lambda|$.
\end{prop}

\begin{proof}
From \eqref{eq:degW} and \eqref{eq:Hlk} it follows that
\[ \deg H_{\lambda, j} = |\lambda|+j-\ell,\quad j\notin \{ k_1,\ldots,
k_\ell \} .\] Hence, if we consider $H_{\lambda, j}$ for $j\geq
k_\ell+1= \lambda_\ell+\ell$ we see that the degree sequence of
$\cU_\lambda$ stabilizes at degree $|\lambda|+\lambda_\ell$, meaning
that $n\in I_\lambda$ for all $n\geq |\lambda|+\lambda_\ell$. In
addition $H_{\lambda, j} \in \cU_\lambda$ for the following set:
\[
j\in\{ \underbrace{0,\dots,k_1-1}_{\lambda_1} \}\cup\{
\underbrace{k_1+1,\dots,k_2-1}_{\lambda_2-\lambda_1} \}\cup\dots\cup
\{ \underbrace{k_{\ell-1}+1,\dots,k_\ell-1}_{\lambda_\ell-\lambda_{\ell-1}} \}
\]
The number of all these additional polynomials is clearly 
\[\lambda_1+(\lambda_2-\lambda_1)+\dots+(\lambda_\ell-\lambda_{\ell-1})
=\lambda_\ell\] and therefore the codimension of the subspace is
$|\lambda|+\lambda_\ell-\lambda_\ell=|\lambda|$.
\end{proof}

The fact that $|\lambda|$ equals both the codimension of $\cU_\lambda$
and the degree of $H_\lambda$ allows us to give the following
characterization of $\cU_\lambda$.  
\begin{prop}
  \label{prop:flagcondition}
  A polynomial $p\in \cP$ belongs to the exceptional  subspace
  $\cU_\lambda$ if and only if
  \begin{equation}
    \label{eq:Tlambdasing}
    2 H_\lambda' (xp-p') + H_\lambda'' p
  \end{equation}
  is divisible by $H_\lambda$.
\end{prop}
\begin{proof}
  The forward implication is true because $T_\lambda$, as defined in
  \eqref{eq:Tlambdadef}, transforms every element of $\cU_\lambda$
  into a polynomial, and because the expression in
  \eqref{eq:Tlambdasing} above is just the numerator of the singular
  part of $T_\lambda$.  For the converse, we observe that divisibility
  of \eqref{eq:Tlambdasing} by $H_\lambda$ imposes precisely \[\deg
  H_\lambda = |\lambda|\] independent linear conditions on $p\in \cP$.
  Since $|\lambda|$ is also the codimension of $\cU_\lambda$ it
  follows by dimensional exhaustion that a polynomial $p\in \cP$ that
  satisfies these conditions is necessarily an element of
  $\cU_\lambda$.
\end{proof}

We now come to the following important question: is the subspace
$\cU_\lambda$ primitive or do the polynomials $H_{\lambda, j}\colon
j\notin \{k_1,\ldots, k_\ell\}$
share a common root? Let us note that the possibility that
$\cU_\lambda$ is imprimitive does not limit the scope of Theorem
\ref{thm:trivmonod}.  By Proposition \ref{prop:Tlambdaeigenvalue}
$T_\lambda$ is solvable by polynomials.  If the elements of
$\cU_\lambda$ share a common factor $\sigma(z)$, then the gauge
transformation
\[ \hT_\lambda = \sigma^{-1}\circ T_\lambda\circ \sigma \] gives a
``reduced'' operator that has polynomials $\sigma^{-1}H_{\lambda, k} $
as eigenfunctions.  By construction, $\hT_\lambda$ is primitive and
exactly solvable by polynomials, and therefore subject to the argument
of Proposition \ref{lem:trivmonod0}.  In this regard, we have the
following characterization of primitivity.
\begin{prop}
  \label{prop:Ulambdaprimitive}
  The subspace $\cU_\lambda$ is primitive if and only if $H_\lambda$ has
  simple roots.
\end{prop}
\begin{proof}
    We define
  \begin{align}
    U_\lambda &= x^2 - 2 \partial_{xx} \log H_\lambda = x^2 + 2 \left(
      \frac{H_\lambda'}{H_\lambda}\right)^2  -\frac{2H_\lambda''}{H_\lambda}\\
    \cH_\lambda[y] &= -y'' + U_\lambda y,\\
    \psi_{\lambda, k} &= e^{-x^2/2}
    \frac{H_{\lambda, k}}{H_\lambda},\quad k\notin \{ k_1,\ldots, k_\ell\}.
  \end{align}
  Then, by a straight-forward calculation, we have
  \[ T_\lambda = -\left(e^{x^2/2} H_\lambda\right)\circ
  \cH_\lambda\circ \left( e^{-x^2/2} H_\lambda^{-1}\right)+1. \] In
  this way, the eigenpolynomial relation \eqref{eq:Tlambdaeigenvalue}
  corresponds to the Hamiltonian eigenrelation
  \[ \cH_\lambda[\psi_{\lambda, k}] = (2k-2\ell+1) \psi_{\lambda, k}.\]

  Writing
  \[ H_\lambda = C \prod_a (x-\xi_a)^{n_a},\quad C\in \Rset, \] where
  the $\xi_a\in \Cset$ are the roots of $H_\lambda$, the $n_a\in
  \Nset$ are the corresponding multiplicities, and where
  \[ \sum_a n_a = \deg H_\lambda = |\lambda|\,, \]
  we have  by a direct calculation
  \begin{equation}
    \label{eq:Tlambdaexpansion}
    T_\lambda[y] =  y'' -2 \left(x + \sum_a \frac{n_a}{x-\xi_a}
    \right) y' + \left(2|\lambda|-1 + \sum_a
      \frac{2r_a}{x-\xi_a} + \sum_a \frac{n_a(n_a-1)}{(x-\xi_a)^2} \right) y\,,
  \end{equation}
  where
  \begin{equation}
    \label{eq:radef}
     r_a = n_a \left( \xi_a - \sum_{b\neq a}
    \frac{n_b}{\xi_b-\xi_a} \right)    
  \end{equation}
  The singularities of the eigenvalue differential equation $T[y] = E
  y$ are regular.  The indicial equation for the pole at $\xi_a$ is
  \[ r^2-r - 2r n_a + n_a(n_a-1) =0\] Since $\cH_\lambda$ has trivial
  monodromy each of the
  \[ n_a = \frac{m_a(m_a+1)}{2},\quad m_a \in \Nset\] is a triangular
  number \cite{Oblomkov1999}.  Applying the above substitution to the
  indicial equation gives
 \begin{equation}
   \label{eq:Tlambdachareq}
   \left(r -\frac{m_a(m_a-1)}{2}\right)\left(r-
     \frac{(m_a+1)(m_a+2)}{2}\right) = 0.
 \end{equation}
 Therefore all power series solutions of $T_\lambda[y]= Ey$ have
 $\prod_a (x-\xi_a)^{m_a(m_a-1)/2}$ as a common factor.  Since 
 $H_{\lambda, k}$ are eigenpolynomials of $T_\lambda$, if $n_a>1$, then
 $m_a>1$, and hence $H_{\lambda, k}(\xi_a) = 0$ for every $k$.

 Conversely, suppose that $\xi\in \Cset$ is a common root of order
 $n>0$ of $H_{\lambda, k}$ for every $k$.  By Lemma
 \ref{lem:ordergaps}, the leading order term in the Laurent expansion
 of $T=T_\lambda$ at $x=\xi$ is
 \[ T_{-2} = \partial_{xx} - \frac{2(\nu+n)}{x-\xi} \partial_x +
 \frac{n(1+n+2\nu)}{(x-\xi)^2} \] for some integer $\nu \geq
 0$. Consequently $x=\xi$ must be a singular point of $T_\lambda$;
 hence $\xi=\xi_a$ for some $a$.  
Direct comparison with
 \eqref{eq:Tlambdaexpansion} shows that
 \[ \nu=m_a,\quad n = \frac{m_a(m_a-1)}{2}.\]
 Since $n>0$ by assumption, $m_a>1$ and hence $n_a>1$ as was to be shown.
\end{proof}

If $\lambda$ is an Adler partition, the subspace $\cU_\lambda$ is primitive provided the recent conjecture put forward by Felder et al. in
\cite{Felder2012a} holds.
\begin{conj}
  All zeros of $H_\lambda$ are simple except possibly for $x=0$.
\end{conj}
\noindent 
For an Adler partition $\lambda$, the zeros of $H_{\lambda}(x)$ do not
lie on the real line, and therefore in this case, the conjecture
together with Proposition \ref{prop:Ulambdaprimitive} imply that
$\cU_\lambda$ is primitive.  In addition, if $\cU_\lambda$ is primitive
it can be characterized in the following, simpler manner, providing explicit first order constraints for its elements at each of the simple roots of $H_\lambda$.

\begin{prop}
  \label{prop:flagwithsimpleroots}
  Suppose that $H_\lambda$ has simple roots, let $\xi_a,\;
  i=1,\ldots, |\lambda|$ be an enumeration of these roots, and let
  \[ r_a = \xi_a + \sum_{b\neq a} \frac{1}{\xi_a-\xi_b}.\] Then $p\in
  \cU_\lambda$ if and only if
  \begin{equation}
    \label{eq:flagcondition}
    p'(\xi_a) - r_a p(\xi_a) = 0,\quad a=1,\ldots, |\lambda|.
  \end{equation}
\end{prop}
\begin{proof}
  The expression in \eqref{eq:flagcondition} follows by a direct
  calculation from evaluating \eqref{eq:Tlambdasing} at a root
  $\xi_a$.  The gist of this calculation  is given above in
  \eqref{eq:Tlambdaexpansion} and \eqref{eq:radef}.
\end{proof}

\subsection{Orthogonality}

Orthogonality of the $\Xl$-Hermite polynomials is a direct consequence of their Sturm-Liouville character. Given
a partition $\lambda$ and the corresponding gap sequence
$k_1,\dots,k_l$ defined by \eqref{eq:klambda}, let us define the
following polynomial
\begin{equation}
  \label{eq:taudef} p_\lambda(x) = (x-k_1) (x-k_2) \cdots (x-k_\ell)
\end{equation} We observe that Adler partitions are precisely those
for which \mbox{$p_\lambda(n)\geq 0$} for all integers $n$. This form of
the  definition was employed  in \cite{Krein1957}.
\begin{prop} The $\Xl$-Hermite polynomials $H^{(\lambda)}_j$ satisfy
the following orthogonality relation
\begin{equation}\label{eq:orthogonality} \int_{-\infty}^\infty
H^{(\lambda)}_i H^{(\lambda)}_j \, W_{\lambda^2}(x)\,dx = \delta_{i,j}
\sqrt{\pi}\, 2^j\, j! \,p_{\lambda^2}(j)
 \end{equation} where the orthogonality weight is given by
\begin{equation}
  \label{eq:Wdefine} W_{\lambda^2}(x)=
\frac{e^{-x^2}}{\left(H_{\lambda^2}(x)\right)^2}.
\end{equation}
\end{prop} Note that the weight $W_{\lambda^2}$ is regular and positive
definite because $\lambda^2$ is an Adler sequence, as it is clear from
\eqref{eq:taudef} and Theorem \ref{thm:krein}.

\begin{proof} The key to the proof is the formal adjoint relation
between the $B$ and $A$ operators defined above.  To be more precise,
we have
\begin{equation}
  \label{eq:ABadjoint} \int^x A_{\lambda, k}[f] g\, W_{\lambda, k}+
\int^x f B_\lambda[g] \, W_\lambda = \frac{e^{-x^2}}{H_\lambda
H_{\lambda, k}}
\end{equation} Iterating the above relation and applying the classical
formula
\begin{equation}
  \label{eq:Hnorm} \int_{-\infty}^\infty H_i(x) H_j(x) e^{-x^2} dx =
\delta_{i,j}\, \sqrt{\pi} \,2^j\, j!
\end{equation} gives the desired generalization
\eqref{eq:orthogonality}.

\end{proof}

\subsection{Completeness}

We have already shown that $ H^{(\lambda)}_j$ for \mbox{$j\in\mathbb
N\backslash \{k_1,k_1+1,\dots,k_\ell,k_\ell+1 \}$} is a numerable sequence of polynomial eigenfunctions of a second order differential operator that satisfies an orthogonality relation. We will now prove that the sequence spans a complete basis of the corresponding Hilbert space. 

In the case of an Adler partition $\lambda$, we know that $|\lambda|=2m$ is an even positive integer and from Theorem \ref{thm:krein} we see that $H_\lambda>0$ for all $x\in\mathbb R$.  It follows that the
weight
\[ W_\lambda(x) = \frac{e^{-x^2}}{H_\lambda(x)^2} dx,\quad x\in \Rset
\] is regular and has finite moments of all orders.
\begin{prop}\label{prop:density}
  If $\lambda$ is an Adler partition, the polynomial subspace $\cU_\lambda$ is dense in the Hilbert space
  $\mathrm{L}^2(\Rset,W_\lambda)$.
\end{prop}

In order to prove the completeness of exceptional Hermite polynomials we will need to preliminary lemmas.
Let $\cH_\alpha=\rL^2[(0,\infty),y^\alpha e^{-y} dy]$ denote the Hilbert
space of the classical Laguerre polynomials and $\cH =
\rL^2[\Rset,e^{-x^2} dx]$ the Hilbert space of the Hermite
polynomials. Throughout the proof, we will make use of Theorem 5.7.1 in \cite{Sz}, which asserts that $\cP$, the vector space of univariate polynomials, is
dense in $\cH_\alpha,\; \alpha>-1$ and in $\cH$.  We will write
\[  q(x) \cP(x) = \{  q(x) p(x) \colon p \in \cP \}\]
to denote  a polynomial subspace with a common factor $q(x)$.
\begin{lem}
  \label{lem:density1}
  The polynomial subspace $(1+y^m)\cP(y)$ is dense in $\cH_\alpha$ for
  every integer $m>0$ and every real $\alpha>0$.
\end{lem}
\begin{proof}
  Given a polynomial $q(y)$ and an $\epsilon>0$ it suffices to find a
  polynomial $p(y)$ such that
  \[ \Vert q- (1+y^m)p \Vert_{\cH_\alpha}^2 = \int_0^\infty \left(
    q(y) - (1+y^m) p(y)\right)^2 y^\alpha e^{-y} dy \leq
  \epsilon. \] Define the function
  \[ \hat{q}(y) =
  \begin{cases}
    \displaystyle \frac{q(y-1)}{(y-1)^m+1}, &  y\geq 1\\[8pt]
    0 & 0\leq y<1
  \end{cases}
  \]
  We assert that $\hat{q} \in \cH_{2m+\alpha}$ by observing that
  \begin{align*}
    \int_0^\infty \hat{q}(y)^2y^{2m+\alpha} e^{-y}\, dy &=
    \int_1^\infty
    \left(\frac{q(y-1)}{(y-1)^m+1}\right)^2 y^{2m+\alpha} e^{-y}\,  dy \\
    & = e^{-1} \int_0^\infty q(y)^2
    \left(\frac{(1+y)^{m}}{1+y^m}\right)^2 (1+y)^\alpha e^{-y} dy\\
    &< \infty
  \end{align*}
  Next, choose a polynomial $\hat{p}(y)$ such that
  \[ \Vert \hat{q} - \hat{p} \Vert_{\cH_{2m+\alpha}}^2 \leq
  \frac{\epsilon}{e}, \] and set
  \[ p(y) = \hat{p}(y+1).\] It follows that
  \begin{align*}
    &\int_0^\infty \left(q(y)- (1+y^m) p(y)\right)^2 y^\alpha
    e^{-y} dy \\
    &\quad = \int_0^\infty \left(\frac{q(y)}{(1+y^m)}-
      p(y)\right)^2 (1+y^m)^2 y^\alpha e^{-y} dy\\
    &\quad \leq \int_0^\infty \left(\frac{q(y)}{1+y^m}-
      p(y)\right)^2 (1+y)^{2m} y^\alpha e^{-y} dy\\
    &\quad = e\int_1^\infty \left(\hat{q}(y)- \hat{p}(y)\right)^2
    y^{2m+\alpha} \left(1-\frac{1}{y}\right)^\alpha
    e^{-y} dy\\
    &\quad \leq e\int_1^\infty \left(\hat{q}(y)- \hat{p}(y)\right)^2
    y^{2m+\alpha}
    e^{-y} dy\\
    &\quad \leq e \Vert \hat{q} - \hat{p}\Vert_{\cH_{2n+\alpha}}\\
    &\quad \leq \epsilon,
  \end{align*}
  as was to be shown.
\end{proof}

\begin{lem}
  \label{lem:density2}
  The polynomial subspace $(1+x^{2m}) \cP(x)$ is dense in $\cH$ for
  every integer $m>0$.
\end{lem}
\begin{proof}
  Given a polynomial $q(x)$ and an $\epsilon>0$ it suffices to find a
  polynomial $p(x)$ such that
  \begin{equation}
    \label{eq:d2qpepsilon}
    \Vert q- (1+x^{2m})p \Vert_{\cH}^2 = \int_\Rset
    \left( q(x) - (1+x^{2m}) p(x)\right)^2 e^{-x^2} dx \leq
    \epsilon.   \end{equation}  
  Write
  \[ q(x) = q_0 + x q_1(x^2) + x^2 q_2(x^2) \] where $q_0$ is a
  constant and where $q_1(y), q_2(y)$ are polynomials in $y=x^2$.
  Imposing the condition that $p(0)=q(0)$ let us write
  \[ p(x) = q_0+ x p_1(x^2)+ x^2 p_2(x^2) \] where $p_1(y), p_2(y)$
  are polynomials. Then, by the orthogonality of odd and even
  functions in $\cH$, the inequality \eqref{eq:d2qpepsilon} assumes
  the form
  \begin{align*}
    &\int_\Rset
    \left( q_1(x^2) - (1+x^{2m})p_1(x^2)\right)^2 x^2e^{-x^2} dx \\
    &\qquad + \int_\Rset \left(q_2(x^2) - q_0 x^{2(m-1)}-
      (1+x^{2m})p_2(x^2)\right)^2 x^4 e^{-x^2} dx\\
    &= \int_0^\infty (q_1(y) - (1+y^m) p_1(y))^2 y^{\frac{1}{2}} e^{-y} dy\\
    &\qquad + \int_0^\infty \left(q_2(y) - q_0 y^{m-1} -
      (1+y^m) p_2(y)\right)^2 y^{\frac{3}{2}} e^{-y} dy \\
    &\leq \epsilon,
  \end{align*}
  where for the first equality we employ the change of variables
  $y=x^2$.  By Lemma \ref{lem:density1}, it is possible to find
  polynomials $p_1(y), p_2(y)$ such that the above inequality is
  satisfied.
  \end{proof}

\begin{proof}[Proof of Proposition \ref{prop:density}]
Let $f\in L^2(\Rset, W_\lambda)$.  
Set
\[ \hat{f}(x) = \frac{(1+x^{2m})}{H_\lambda(x)^2}f(x)
\] and observe that
\begin{align*}
 \int_\Rset \hat{f}(x)^2 e^{-x^2}\, dx 
  &\leq A^2 \int_\Rset \left(\frac{f(x)}{H_\lambda(x)}\right)^2
  e^{-x^2}\,dx = \int_\Rset f(x)^2 W_\lambda(x) \,dx < \infty
\end{align*}
where 
\[ A = \sup \left\{ \frac{1+x^{2m}}{H_\lambda(x)} \colon x \in \Rset
\right\}<\infty .\] Let $\epsilon>0$ be given.  
Set
\[ B = \sup \left\{ \frac{H_\lambda(x)}{1+x^{2m}} \colon x \in \Rset
\right\}<\infty.\] By Lemma \ref{lem:density2} we can find a
polynomial $p(x)$ such that
\[ \int_\Rset \left( \hat{f}(x)- (1+x^{2m}) p(x)\right)^2 e^{-x^2}\, dx
\leq \frac{\epsilon}{B^2}.\] Hence
\begin{align*}
  & \int_\Rset \left( f(x) - H_\lambda(x)^2 p(x) \right)^2 W_\lambda(x)\, dx\\
  &\quad = \int_\Rset \left( \frac{f(x)}{H_\lambda(x)} - H_\lambda(x)
    p(x) \right)^2 e^{-x^2}\, dx\\
  &\quad \leq B^2 \int_\Rset \left( \hat{f}(x) - (1+x^{2m})
    p(x) \right)^2 e^{-x^2}\, dx\\
  &\quad \leq \epsilon
\end{align*}
By Proposition \ref{prop:flagcondition} a polynomial of the form
$H_\lambda(x)^2 p(x)$ belongs to $\cU_\lambda$.   The Proposition is proved.
\end{proof}

\subsection{Recursion formulas}

Although the definition of $\Xl$-Hermite polynomials via a Wronskian determinant of classical Hermite polynomials is the most compact way of defining them, it is clear that for the purpose of an efficient computation it would be better to have a recursion formula. The existence of a recursion formula for the exceptional polynomial families has been a major challenge in the past few years until the recent work of Odake \cite{Odake2013}, which shows an elegant way to derive such recurrence relations for certain families of exceptional polynomials. 

Although the recursion formulas given in \cite{Odake2013} are given for sequences of exceptional polynomials whose degree sequence is $m,m+1,m+2,\dots$, i.e. with one gap at the beginning, the procedure can also be applied to sequences with an arbitrary degree sequence.
In the case of Hermite polynomials, in addition, the recursion for the connection polynomials (which are only defined recursively in \cite{Odake2013})  can be solved explicitly.

For integers $n,i,j$, let
\[ C^n_{ij} =
\begin{cases}
  \displaystyle \frac{n!}{i! j! (n-i-j)!},&\quad  i,j,n-(i+j)\geq 0\\
  0, &\quad \text{otherwise}
\end{cases}
\]
be the trinomial coefficient, and for an integer $i$ let

\[ (x)_i =
\begin{cases}
  x (x+1)\cdots (x+i-1),& i> 0\\
  1,& i=0\\
  0,& i<0
\end{cases}
\]
denote the usual Pochhammer symbol.  
In the rest of the section, we adopt the convention that $H_{\lambda,n}=0$ for $n<0$.
\begin{prop}\label{prop:recur}
  Let $\lambda=(\lambda_1,\dots,\lambda_\ell)$ be a partition of length $\ell$ and $n \geq -\ell-1$ be a non-negative integer.  The Wronskians of Hermite polynomials $H_{\lambda, k}$ defined by \eqref{eq:Hlk} obey the following  relation:
  \begin{equation}
    \label{eq:xhrecur}
    0=\sum_{j,k,m\atop {2j-m-k=0}} (-1)^m 2^{-j} (n+k+1)_{\ell+1-j}\,
    C^{\ell+1}_{m,j-m}\, H_{m}H_{\lambda,n+k} ,\qquad n\geq-\ell-1
  \end{equation}
  where the sum is taken over all non-negative integers $k,m,j$ that
  satisfy the above constraint.
\end{prop}

The following remarks are in place:

\begin{enumerate}

\item The sum \eqref{eq:xhrecur} is finite because the Pochhammer symbol evaluates
  to zero if $j>\ell+1$ and the trinomial symbol evaluates to zero if
  $j>\ell+1$ or if $m>j$. 
   In effect, the above identity can be expressed in the following manner
   \begin{subequations}
   \begin{equation}
   \sum_{k=0}^{2\ell+2} B^\ell_{n,k} H_{\lambda,n+k},\quad n\geq -\ell-1,
   \end{equation}
   where
   \begin{equation}
   B^\ell_{n,k} =\sum_{j=\lceil\frac{k}{2}\rceil}^{\min(k,\ell+1)} (-1)^k 2^{-j} (n+k+1)_{\ell+1-j} C^{\ell+1}_{2j-k,k-j} H_{2j-k}
   \end{equation}
   \end{subequations}
   is a linear combination of even(odd) Hermite polynomials if $k$ is even(odd).
  \item Since  the index $k$ ranges from $0$ to $2(\ell+1)$, we have
  a $(2\ell+3)$-term recurrence relation for the $(\ell+1)$-order Wronskians.
\item The recursion formula \eqref{eq:xhrecur} is valid for Wronskians of an arbitrary sequence of Hermite polynomials, whether they form an Adler sequence or not. 
\item The initial values  of the recurrence relation are given by
\[H_{\lambda,0}, H_{\lambda,1},\dots,H_{\lambda,\ell}\]
where at least one of which has to be non-zero since the length of $\lambda$ is $\ell$ and that is the size of the largest possible gap.  The starting equation of the recurrence occurs at $n=-\ell-1$ and it allows to compute $H_{\lambda,\ell+1}$.

\end{enumerate}

Let us illustrate these recursion formulas by providing the first few explicit examples.

\begin{itemize}
\item 
When $\ell=0$, the above identity reduces to the usual 3-term recurrence relation for Hermite polynomials, which we express as
\[ \frac{1}{2} H_0 H_{n+2} - \frac{1}{2} H_1 H_{n+1} + (n+1) H_0 H_n =
0,\quad n\geq -1 \]

\item For $\ell=1$, the second order Wronskians satisfy the following 5-term recurrence relation

\begin{align*}
  0=&4(n+1)_2 H_0 H_{\lambda,n}-4 (n+2) H_1 H_{\lambda,n+1} + \Big( 4(n+3) H_0 +
    H_2\Big) H_{\lambda,n+2} \\
  &\quad - 2 H_1 H_{\lambda,n+3}+ H_0  H_{\lambda,n+4}  ,\qquad n\geq -2.
\end{align*}
The initial values for this sequence are $H_{\lambda,0}$ and $H_{\lambda,1}$,  and one of them has to be nonzero since $\lambda=(\lambda_1)$ has length $\ell=1$. At $n=-2$ the above recurrence allows to compute $H_{\lambda,2}$.

\item For $\ell=2$, $\lambda=(\lambda_1,\lambda_2)$ and the third order Wronskians $H_{\lambda, k}$ satisfy the following 7-term recurrence relation
\begin{align*}
  0&=8(n+1)_3\, H_0 \,H_{\lambda,n}-12 (n+2)_2 \,H_1\,
  H_{\lambda,n+1}\\
  &\quad +\Big(12 (n+3)_2\, H_0+6(n+3)\,    H_2\Big)
  \,H_{\lambda,n+2}\\
  &\quad    -\Big(12 (n+4)\,
    H_1+H_3\Big) H_{\lambda,n+3}\\
  &\quad +\Big(6(n+5)\,
    H_0+ 3H_2\Big) H_{\lambda,n+4}\\
  &\quad -3 H_1\,
  H_{\lambda,n+5}+ H_0\, H_{\lambda,n+6},\qquad n\geq-3.
\end{align*}
The initial values for this sequence are $H_{\lambda,0},H_{\lambda,1},H_{\lambda,2}$,  and one of them has to be nonzero since $\lambda=(\lambda_1,\lambda_1)$ has length $\ell=2$. At $n=-3$ the above recurrence allows to compute $H_{\lambda,3}$.
\end{itemize}
\begin{proof}[Proof of Proposition \ref{prop:recur} ]
For a given partition $\lambda=(\lambda_1,\dots,\lambda_\ell)$ of length $\ell$  we observe that the coefficients of the recurrence relation depend only on the partition length. The proof follows the argument given by Odake in \cite{Odake2013} and proceeds by induction on the number $\ell$ of Darboux transformations, which motivates the introduction of the following shorthand notation
\begin{align*}
  H^{(\ell)} &= H_\lambda= \Wr[H_{k_1},\ldots, H_{k_\ell}] \\
  H^{(\ell)}_{n} &= H_{\lambda,n}=\Wr[H_{k_1},\ldots, H_{k_\ell},H_n] \\
  H^{(\ell+1)} &= \Wr[H_{k_1},\ldots, H_{k_\ell},H_{k_{\ell+1}}] \\
  H^{(\ell+1)}_{n} &= \Wr[H_{k_1},\ldots, H_{k_\ell},H_{k_{\ell+1}},H_n] 
\end{align*}
The recurrence relation \eqref{eq:xhrecur} can be then expressed as
\begin{equation}
  \label{eq:xrecur1}
  \sum_{j=0}^{\ell+1} \sum_{m=0}^j A^{\ell+1}_{njm}\,
  H_{m}H^{(\ell)}_{n+2j-m} = 0,
\end{equation}
where
\[ A^\ell_{njm} = (-1)^m 2^{-j} (n+2j-m+1)_{\ell-j}\,
C^{\ell}_{m,j-m}. \]

The proof will be completed in two steps. First we will need to prove the following Lemma.

\begin{lem}\label{lem:rec}
The following relation 
\begin{equation}
  \label{eq:xrecur2}
  \sum_{j=0}^{\ell} \sum_{m=0}^j
  A^\ell_{njm}\, H_{m}H^{(\ell)}_{n+2j-m} =    \ell!\,  H_{n+\ell} H^{(\ell)}
\end{equation}
implies \eqref{eq:xrecur1}.
\end{lem}
\begin{proof}
  Suppose that \eqref{eq:xrecur2} holds.  We multiply the ordinary
  3-term recurrence for the classical Hermite polynomials
  \[ 2(\ell+n-1) H_{n+\ell-2} -2x H_{n+\ell-1} + H_{n+\ell} =0\]
  by $\ell!  H^{(\ell)}$ to obtain
  \begin{align*}
    0&= \sum_{j=0}^\ell \sum_{m=0}^j \left(2(\ell+n-1)
      A^\ell_{n-2,j,m} H^{(\ell)}_{n-2+k} - 2x A^\ell_{n-1,j,m}
      H^{(\ell)}_{n-1+k}+
      A^\ell_{n,j,m} H^{(\ell)}_{n+k}\right) H_m\\
    &= \sum_{j=0}^\ell \sum_{m=0}^j \left(2(\ell+n-1) A^\ell_{n-2,j,m}
      H^{(\ell)}_{n-2+k} +
      A^\ell_{n,j,m} H^{(\ell)}_{n+k}\right) H_m \\
    &\quad - \sum_{j=0}^\ell \sum_{m=0}^j  \left(H_{m+1}+ 2m H_{m-1}\right)
    A^\ell_{n-1,j,m} H^{(\ell)}_{n-1+k},
  \end{align*}
  where $k=2j-m$ for notational convenience, and where we use the
  3-term recurrence again in the second equation.  Reindexing the
  above sums gives
  \begin{align*}
    0&= \sum_{j=0}^\ell \sum_{m=0}^j 2(\ell+n-1) A^\ell_{n-2,j,m}
    H^{(\ell)}_{n-2+k} H_m + \sum_{j=1}^{\ell+1} \sum_{m=0}^{j-1}
    A^\ell_{n,j-1,m} H^{(\ell)}_{n-2+k} H_m \\
    &\quad - \sum_{j=1}^{\ell+1} \sum_{m=1}^{j} A^\ell_{n-1,j-1,m-1}
    H^{(\ell)}_{n-2+k}H_{m} - \sum_{j=0}^\ell \sum_{m=0}^{j-1} 2(m+1)
    A^\ell_{n-1,j,m+1} H^{(\ell)}_{n-2+k}H_{m}.
  \end{align*}
  By the definition of the trinomial $C^\ell_{j,j-m}$ symbol, 
  \begin{align*}
    &A^\ell_{n-2,j,m}=0 & \text{if } &  j=\ell+1,\\
    &A^\ell_{n,j-1,m} = 0 & \text{if }&   m=j\quad\text{or } j=0,\\
    &A^\ell_{n-1,j-1,m-1} =0& \text{if }&   m=0\quad\text{or } j=0,\\
    &A^\ell_{n-1,j,m+1}  = 0 & \text{if } & m=j\quad\text{or } j=\ell+1.
  \end{align*}
  Hence all of the above sums can be collected as a single sum with
  the same range indices; namely,
  \begin{align*}
    0 &=\sum_{j=0}^{\ell+1} \sum_{m=0}^j 
    \left( 2(\ell+n-1) A^\ell_{n-2,j,m}+ A^\ell_{n,j-1,m}- A^\ell_{n-1,j-1,m-1}- 2(m+1)
    A^\ell_{n-1,j,m+1} \right) H^{(\ell)}_{n-2+k}H_{m}.
  \end{align*}
  An elementary calculation now gives
  \[ 2(\ell+n-1) A^\ell_{n-2,j,m}+ A^\ell_{n,j-1,m}-
  A^\ell_{n-1,j-1,m-1}- 2(m+1) A^\ell_{n-1,j,m+1} =
  2A^{\ell+1}_{n-2,j,m} ,\]
  thereby giving \eqref{eq:xrecur1}, albeit with $n\to n-2$  
\end{proof}

  We now use Lemma \ref{lem:rec} to establish \eqref{eq:xrecur2} by induction on
  $\ell$.  Specializing \eqref{eq:xrecur1} to the case $\ell=0$
  yields,
  \begin{align*}
    0 &= A^1_{n00} H_0 H_n + A^1_{n01} H_0
    H_{n+2} +A^1_{n10} H_1 H_{n+1}\\
    &= (n+1) H_n + 2^{-1} H_{n+2} - x H_{n+1},
  \end{align*}
  which is the usual recurrence relation.

  We now suppose that \eqref{eq:xrecur2} and hence \eqref{eq:xrecur1} hold
  for a given $\ell$.  The first step is to form the Wronskian of
  \eqref{eq:xrecur1} with $H^{(\ell+1)}$.  Using the generalized
  Leibnitz identity
    \[ \Wr[f,gh]  = h \Wr[f,g] + h' f g \]
    and the Hermite lowering relation
    \[ H_m' = 2m H_{m-1} \]
    gives
    \begin{align*}
      0&=\sum_{j=0}^{\ell+1} \sum_{m=0}^j A^{\ell+1}_{njm}\, H_m
      \Wr[H^{(\ell+1)},H^{(\ell)}_{n+2j-m}] +H^{(\ell+1)}
      \sum_{j=1}^{\ell+1} \sum_{m=1}^j 2m A^{\ell+1}_{njm}\,
      H_{m-1}H^{(\ell)}_{n+2j-m}
    \end{align*}
    Next, we employ the ``Wronskian of  Wronskians'' identity
    \[ \Wr[ \Wr[f_1, \ldots, f_n, g], \Wr[f_1,\ldots, f_n, h]] =
    \Wr[f_1,\ldots, f_n] \Wr[f_1,\ldots, f_n, g,h] \]
    to obtain
    \begin{align*}
      0&=H^{(\ell)} \sum_{j=0}^{\ell+1} \sum_{m=0}^jA^{\ell+1}_{njm}
      \, H_m H^{(\ell+1)}_{n+2j-m} +H^{(\ell+1)} \sum_{j=1}^{\ell+1}
      \sum_{m=1}^j 2m
      A^{\ell+1}_{njm}\, H_{m-1}H^{(\ell)}_{n+2j-m}\\
      &=H^{(\ell)} \sum_{j=0}^{\ell+1} \sum_{m=0}^jA^{\ell+1}_{njm} \,
      H_m H^{(\ell+1)}_{n+2j-m} -(\ell+1) H^{(\ell+1)} \sum_{j=0}^{\ell}
      \sum_{m=0}^j A^{\ell}_{n+1,jm}\, H_{m-1}H^{(\ell)}_{n+2j-m},
    \end{align*}
    where we used the identity
    \[ (1+\ell) A^\ell_{n+1,jm} + 2(m+1) A^{\ell+1}_{n,j+1,m+1} = 0. \]
    Applying \eqref{eq:xrecur2} and cancelling $H^{(\ell)}$  gives
    \begin{equation}
      \label{eq:xrecur3}
      \sum_{j=0}^{\ell+1} \sum_{m=0}^jA^{\ell+1}_{njm} \,
      H_m H^{(\ell+1)}_{n+2j-m} =(\ell+1)! H_{n+\ell+1} H^{(\ell+1)},
    \end{equation}
    which is \eqref{eq:xrecur2} with $\ell\to \ell+1$, thus closing the induction.
\end{proof}

\section{Discussion}

The results of this paper can be viewed from different angles depending on the focus of interest. From the point of view of monodromy free potentials, we have proved that every second order differential operator whose leading order is a constant that admits a numerable sequence of eigenpolynomials is monodromy free. Even more, its general solution is an entire function. This implies in particular that all rational extensions of the quantum harmonic oscillator are monodromy free, and therefore can be obtained by a sequence of rational Darboux transformations from the harmonic potential. 

It is important to note at this point that the harmonic oscillator only admits state-adding and state-deleting rational Darboux transformations, but no isospectral ones, which makes the whole classification easier to tackle. Moreover, as it was proved by Oblomkov in \cite{Oblomkov1999} and later observed by Felder et al. in \cite{Felder2012a}, the use of state-adding Darboux transformations does not lead to any new rational extensions, and any such combination can be expressed in terms of state-deleting transformations alone. This idea has been further extended by Sasaki and Odake in \cite{Odake2013z} to other polynomial families, and to an arbitrary shape invariant potential in \cite{GGM13}.

From the point of view of exceptional orthogonal polynomials, the relevance of this paper is twofold: first, it gives a partial answer to the conjecture formulated in \cite{Gomez-Ullate2012} for the Hermite case: every family of exceptional Hermite polynomials can be obtained from the classical Hermite by a sequence of rational (state-deleting) Darboux transformations. Second, the paper shows the proper way to index the $\Xl$-Hermite polynomials, and it provides some of their properties: differential equations, orthogonality relations, proof of completeness, etc.

Another main novelty consists in providing explicit recurrence relations satisfied by the $\Xl$-Hermite polynomials, by generalizing and explicitly solving the relations first proposed by Odake in \cite{Odake2013}. We observe that the order of the recurrence is $2\ell+3$ where $\ell$ is the number of rational state-deleting Darboux transformations. However, the duality between state-adding and state-deleting transformations implies that there will be recurrence relations of lower order for some of these families. Otherwise speaking, if the lowest order recurrence relation
for a given family of $\Xl$-Hermite polynomials wants to be given, both state-adding and state-deleting Darboux transformations need to be considered, \cite{GGM13}.

\section{Acknowledgements}

The research of the first author (DGU) has been
supported in part by Spanish MINECO-FEDER Grants MTM2009-06973, MTM2012-31714, and the Catalan Grant 2009SGR--859. The research of the third author (RM) was supported in part by NSERC grant RGPIN-228057-2009.

\end{document}